\documentclass[10pt,twocolumn,aps,pra,showpacs,superscriptaddress,floatfix,nofootinbib,longbibliography]{revtex4-1}
\usepackage[T1]{fontenc}
\usepackage{graphicx}
\usepackage{amssymb}
\usepackage{amsmath}
\usepackage{color}
\usepackage{psfrag}
\usepackage{epsfig}
\usepackage{bbm}
\usepackage{bm}
\usepackage[colorlinks=true,linkcolor=blue]{hyperref}
\usepackage[normalem]{ulem}
\usepackage{amsthm}
\usepackage{epstopdf}

\newcommand{\ket}[1]{| #1 \rangle}
\newcommand{\bra}[1]{\langle #1 |}

\newcommand{\beq}{\begin{eqnarray}}
\newcommand{\eeq}{\end{eqnarray}}

\definecolor{nred}{rgb}{0.9,0.1,0.1}
\definecolor{nblack}{rgb}{0,0,0}
\definecolor{nblue}{rgb}{0.2,0.2,0.8}
\definecolor{ngreen}{rgb}{0.2,0.6,0.2}

\newcommand{\SMDI}{{\mathcal{S}^{\text{\tiny MDI}}(\{\sigma_{a|x}\})}}
\newcommand{\SOMDI}{{\mathcal{S}_0^{\text{\tiny MDI}}(\{\sigma_{a|x}\})}}
\newcommand{\EBB}{{E_1^{BB_{0}}}}
\newcommand{\PhiBB}{{\Phi^{BB_{0}}_+}}
\newcommand{\parallelsum}{\mathbin{\!/\mkern-5mu/\!}}

\DeclareMathOperator{\tr}{tr}

\theoremstyle{definition}
\newtheorem{theorem}{Theorem}
\newtheorem{lemma}{Lemma}

\begin{document}

\title{Measurement-device-independent measure of  Einstein-Podolsky-Rosen steering}

\author{Huan-Yu Ku}
\email{huan-yu.ku@riken.jp}
\affiliation{Theoretical Quantum Physics Laboratory, RIKEN Cluster for Pioneering Research, Wako-shi, Saitama 351-0198, Japan}
\affiliation{Department of Physics, National Cheng Kung University, 701 Tainan, Taiwan}
\affiliation{Center for Quantum Frontiers of Research \& Technology, NCKU, 701 Tainan, Taiwan}

\author{Shin-Liang Chen}
\email{shin-liang.chen@mpq.mpg.de}
\affiliation{Max-Planck-Institut f{\"u}r Quantenoptik, Hans-Kopfermann-Stra{\ss}e 1, 85748 Garching, Germany}
\affiliation{Munich Center for Quantum Science and Technology (MCQST), Schellingstra{\ss}e. 4, D-80799 M{\"u}nchen, Germany}
\affiliation{Department of Physics, National Cheng Kung University, 701 Tainan, Taiwan}

\author{Hong-Bin Chen}
\affiliation{Department of Physics, National Cheng Kung University, 701 Tainan, Taiwan}
\affiliation{Center for Quantum Frontiers of Research \& Technology, NCKU, 701 Tainan, Taiwan}

\author{Franco Nori}
\affiliation{Theoretical Quantum Physics Laboratory, RIKEN Cluster for Pioneering Research, Wako-shi, Saitama 351-0198, Japan}
\affiliation{Department of Physics, The University of Michigan, Ann Arbor, 48109-1040 Michigan, USA}

\author{Yueh-Nan Chen}
\email{yuehnan@mail.ncku.edu.tw}
\affiliation{Department of Physics, National Cheng Kung University, 701 Tainan, Taiwan}
\affiliation{Center for Quantum Frontiers of Research \& Technology, NCKU, 701 Tainan, Taiwan}

\date{\today}

\begin{abstract}
Within the framework of quantum refereed steering games (QRSGs), quantum steerability can be certified without any assumption on the underlying state nor the measurements involved. Such a scheme is called measurement-device-independent (MDI) scenario. In this work, we define a variant of QRSGs and introduce a measure of steerability in a MDI scenario, i.e., the result merely depends on the observed data table and the quantum inputs. 
We prove that such a measure is robust against measurement imperfections and show that it is a convex steering monotone by proving the equivalence to the steering fraction as well as the steering robustness. 
Finally, we provide an analytical expression of the measure for a family of two-qubits Werner states. 
\end{abstract}

\maketitle
\emph{Introduction.---}Entanglement~\cite{Einstein35}, steerability~\cite{Schrodinger35}, and Bell nonlocality~\cite{Bell64} are three types of quantum correlations which play essential roles in quantum cryptography, quantum communication, teleportation, and quantum information processing~\cite{Horodecki09RMP,Brunner14RMP,Branciard12}. The fact that Einstein-Podolsky-Rosen (EPR) steering is treated as an intermediate quantum correlation between entanglement and nonlocality leads to a hierarchical relation among them. That is, all nonlocal states are steerable, and all steerable states are entangled, but not vice versa~\cite{Wiseman07,Jones07,Quintino15}. During the past decade, there have been many significant experimental works~\cite{Saunders10,  Bennet12,Handchen12, Smith12,Schneeloch13,Su13,Sun16, Cavalcanti09} and various theoretical results on EPR steering~\cite{Reid89,Pusey13,Walborn11,Kogias15,Costa16}, including the correspondence with measurement incompatibility~\cite{Cavalcanti16,Uola14,Quintino14,Shin-Liang16c,Uola15}, one-way steering~\cite{ Wollmann16,Bowles14}, temporal steering~\cite{Yueh-Nan14,Shin-Liang16,Ku16,Che-Ming15,Ku18b}, continuous-variable EPR steering~\cite{Tatham12,Qiongyi15,Xiang17}, as well as measures of EPR steering~\cite{Skrzypczyk14,Piani15,Hsieh16,Gallego15,Eneet17b,Eneet17a,SDPreview17,Ku18a}.

Bell nonlocality enables one to perform so-called \emph{device-independent} (DI) quantum information processing~\cite{Brunner14RMP,Gallego10,Bancal11,Cavalcanti12,Acin07}, i.e., one makes no assumption on the underlying quantum state nor on the quantum measurements performed. From the hierarchical relation~\cite{Wiseman07}, it naturally leads to the fact that a Bell inequality can be treated as a \emph{DI entanglement witness}. Nevertheless, not all entangled states can be detected by using a Bell inequality violation~\cite{Werner89}. Recently, based on Buscemi's \emph{semi-quantum nonlocal games}~\cite{Francesco12}, Branciard \emph{et al.}~\cite{Branciard13} proposed a collection of entanglement witnesses in the so-called \emph{measurement-device-independent} (MDI) scenario. Compared with the standard DI scenario, there is one more assumption in a MDI scenario: the input of each detector has to be a set of tomographically complete quantum states instead of real numbers. Such a simple relaxation leads to that \emph{all} entangled states can be certified by the proposed MDI entanglement witnesses~\cite{Francesco12,Branciard13}. This characterization gives rise to the recent works providing frameworks for MDI measure of entanglement~\cite{Rosset18b,Shahandeh17,Verbanis16}, non-classical teleportation~\cite{Cavalcanti17}, and non-entanglement-breaking channel verification~\cite{Rosset18}.


Recently, Cavalcanti~\emph{et al.}~\cite{Cavalcanti13} introduced another type of nonlocal game, dubbed as \emph{quantum refereed steering games} (QRSGs). In each of such games, one player is questioned and answers with real numbers, while the other player is questioned with (isolated) quantum states but still answers with real numbers. They showed that there always exists a QRSG with a higher winning probability when the players are correlated by a steerable state~\cite{Cavalcanti13}. Later, Kocsis \emph{et al.}~\cite{Kocsis15} experimentally proposed a QRSG via steering inequality to verify the steerability for the family of two-qubit Werner states in such a scenario, which is also referred to as a MDI scenario.

In this work, we consider a variant of QRSGs, by which we propose the first MDI steering measure (MDI-SM) of the underlying unknown steerable resource without accessing any knowledge of the involved measurement.
We show that \emph{the MDI-SM is a standard measure of steerability, i.e., a convex steering monotone~\cite{Gallego15}}, by proving that it is equivalent to the previous proposed measures: the steering robustness~\cite{Piani15} as well as the steering fraction~\cite{Hsieh16}. Therefore, our proposed measure not only demonstrates the degree of steerability of the underlying steering resource, but also quantifies the degree of entanglement of the sharing quantum state~\cite{Shahandeh17} as well as measurement incompatibility. We note that this is the first time to estimate the degree of measurement incompatibility of the involved measurements in a MDI scenario.
Furthermore, a variant of QRSGs also provides a general method for constructing a collection of MDI steering witnesses for all steerable resources.
We also show that our proposed MDI-SM is robust, in the sense that it is able to detect steerability with detection losses.
Since our approach does not make any assumption on the underlying resource, including the dimension, a high-dimensional MDI-SM is in principle experimentally accessible with current technology in linear optics~\cite{Zhang19}.
The breakthrough of our work is that our proposed measure is the first one which is viable in a MDI scenario, i.e., merely relies on the experimental data table without knowing the full knowledge about the underlying steerable resource, while all the other ones are designed for the standard steering scenario. An explicit example is shown to demonstrate how to implement the proposed measure.



\begin{figure}[tbp]
\includegraphics[width=0.8\columnwidth]{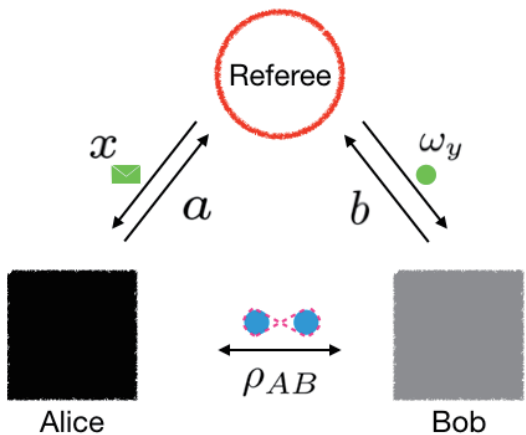}
\caption{Schematic diagrams for the quantum refereed steering games (QRSGs)~\cite{Cavalcanti13}. A QRSG is composed of three spatially separated parties, Alice, Bob, and a referee. Alice and Bob share a quantum state $\rho_{AB}$. The referee encodes the questions within classical variables $\{x\}$ and quantum states $\{\omega_y\}$ to Alice and Bob, respectively. According to the classical inputs $\{x\}$, Alice performs a set of uncharacterized measurements with outcomes $\{a\}$, described by a set of POVMs $\{E_{a|x}\}$. On the other hand, Bob performs a set of uncharacterized joint measurements described by a POVM $\{E_b^{BB_0}\}$ with outcomes $\{b\}$. According to their data table $\{p(a,b|x,\omega_y)\}$, the referee is able to measure the degree of steerability of the underlying unknown steerable resource without trusting Alice's and Bob's measurements. In this work, we consider a variant of QRSGs from the perspective of the resource theory of steering~\cite{Gallego15}, under which the steerable resource is composed of a set of subnormalized states $\{\sigma_{a|x}\}$, called an \emph{assemblage}, produced by Alice's measurements $\{E_{a|x}\}$ acting on the shared state $\rho_{AB}$.
}
\label{fig_EPR_MDI} %
\end{figure}

\emph{Quantum refereed steering games.---}In this work, we assume that all quantum states act on a finite dimensional Hilbert space $\mathcal{H}$. The sets of density matrices and operators acting on $\mathcal{H}$ are denoted by $\mathsf{D}(\mathcal{H})$ and $\mathsf{L}(\mathcal{H})$, respectively.
We denote the index sets of finite number of elements by $\mathcal{A}$, $\mathcal{B}$, $\mathcal{X}$, and $\mathcal{Y}$.
The probability of a specific index, say $a\in\mathcal{A}$, is denoted by $p(a)$. 

A QRSG~\cite{Cavalcanti13} consists of one referee and two players, referred to as Alice and Bob. Besides, Alice and Bob share a quantum state $\rho_{AB}\in \mathsf{D}(\mathcal{H}_{A}\otimes \mathcal{H}_{B})$ and are prohibited to communicate with each other. During each round of the game, Alice receives a classical number $x\in\mathcal{X}$ with probability $p(x)$ as her question from the referee, while Bob receives a quantum state $\omega_y^{B_{0}}\in \mathsf{D}(\mathcal{H}_{B_{0}})$ ($\omega_y$ in short) with probability $p(y)$ (where $y\in\mathcal{Y}$) as his question. To respond to the referee, Alice performs a quantum measurement, described by a POVM $\{E_{a|x}\in \mathsf{L}(\mathcal{H}_{A})\}$, on her part of the system $\text{tr}_B(\rho_{AB})$ and sends the measurement outcome $a\in\mathcal{A}$ as her answer to the referee, while Bob performs a joint quantum measurement, described by a POVM $\{E_{b}^{BB_0}\in\mathsf{L}(\mathcal{H}_{B}\otimes\mathcal{H}_{B_{0}})\}$, on his part of the system $\text{tr}_A(\rho_{AB})$ together with the quantum question $\omega_y$ received from the referee, and sends the output $b\in\mathcal{B}$ as his answer to the referee (see Fig.~\ref{fig_EPR_MDI}). Finally, according to the questions and answers, the referee gives Alice and Bob a payoff $\mu=\mu(a,b,x,\omega_y)$. After many rounds (within the same game), the average payoff they obtain is
\begin{equation}
\tilde{I}(\rho_{AB},\{\mu\})=\sum_{a,b,x,y}p(x)p(y)\mu(a,b,x,\omega_y)p(a,b|x,\omega_y),
\label{payoff}
\end{equation}
where
$p(a,b|x,\omega_y)=\tr\left[ (E_{a|x}\otimes E_b^{BB_{0}})(\rho_{AB}\otimes \omega_y ) \right]$
is the probability distribution of Alice's and Bob's answers according to the questions they receive for all $a,b,x,y$. It was shown that any steerable state allows them to obtain a higher value of the average payoff than the one from unsteerable states~\cite{Cavalcanti13}.

\emph{MDI measure of steerability.---}In this section, we consider a variant of QRSGs.
This will be shown to be helpful in introducing the MDI steering measure (MDI-SM) without knowing the full knowledge about a steerable resource.
We then show it is a standard measure of steerability, i.e., a convex steering monotone~\cite{Gallego15}, by proving that it is equivalent to the steering robustness~\cite{Piani15} and the steering fraction~\cite{Hsieh16}. We also show that the MDI-SM is robust against detection losses.
We stress that the main difference between the standard steering scenario and the MDI scenario is that, the former needs the full knowledge about the underlying steerable resource, while the latter is merely based on the observed statistics $\{p(a,b|x,\omega_y)\}$.

Under the framework of the resource theory of EPR steering~\cite{Gallego15}, the correlation is obtained by Bob's joint measurement $\{E_b^{BB_{0}}\}$ acting on an \emph{assemblage} $\{\sigma_{a|x}\}$ via
\begin{equation}
p(a,b|x,\omega_y)=\text{tr}(E_b^{BB_{0}} \sigma_{a|x}\otimes\omega_y).
\label{pabxy}
\end{equation}
An assemblage $\{\sigma_{a|x}\}$ is a set of subnormalized quantum states defined by $\sigma_{a|x}=\text{tr}_A(\rho_{AB}~E_{a|x}\otimes\openone)$~\cite{Pusey13}, which includes both the information of Alice's marginal statistics $p(a|x)=\text{tr}(\sigma_{a|x})$ and the normalized states $\hat{\sigma}_{a|x}=\sigma_{a|x}/p(a|x)\in \mathsf{D}(\mathcal{H}_{B})$ Bob receives. An assemblage is said to be \emph{unsteerable} if it admits a local-hidden-state (LHS) model~\cite{Wiseman07}: $\sigma_{a|x}=\sigma^{\text{US}}_{a|x}=\sum_{\lambda}p(\lambda)p(a|x,\lambda)\sigma_\lambda~\forall~a,x;$ otherwise, it is steerable. In particular, the set of all unsteerable assemblages $\mathsf{LHS}$ forms a convex set~\cite{Skrzypczyk14,Gallego15}; consequently, for a given steerable assemblage $\{\sigma_{a|x}^{\rm{S}}\}$, there always exists a set of positive semidefinite operators $\{F_{a|x}\geq 0\}$, called a \emph{steering witness}, such that $\tr\sum_{a,x}F_{a|x}\sigma_{a|x}^{\rm{S}} > \alpha$, while $\tr\sum_{a,x}F_{a|x}\sigma_{a|x}^{\rm{US}} \leq \alpha \quad \forall \{\sigma_{a|x}^{\rm{US}}\}\in\mathsf{LHS}$~\cite{Cavalcanti09,Pusey13,SDPreview17,Skrzypczyk14,Piani15}, where $\alpha:= \max_{\{\sigma_{a|x}^{\rm{US}}\}\in\mathsf{LHS}} \tr \sum_{a,x} F_{a|x}\sigma_{a|x}^{\rm{US}}$ is the local bound of the steering witness.

In what follows, we will use the property of the existence of a steering witness for a steerable assemblage to construct the MDI-SM.
First, we map a QRSG to a set of real numbers $\beta=\{\beta_{a,1}^{x,y}\}$ by choosing the relation $\mu(a,b,x,y)=\beta^{x,y}_{a,b}\,\delta_{1,b}/[p(x)p(y)]$~\cite{Shahandeh17}. Therefore, the average payoff of such a variant of a QRSG can be written as
\begin{equation}
I\left(\mathbf{P},\beta\right)=\sum_{a,x,y}\beta^{x,y}_{a,1}p(a,1|x,\omega_y),
\label{payoff2}
\end{equation}
where $\mathbf{P}:= \{p(a,1|x,\omega_y)\}$ is the observed statistics from an unknown assemblage $\{\sigma_{a|x}\}$ according to Eq.~\eqref{pabxy}.
One notes that the average payoff in Eq.~\eqref{payoff2} cannot only be seen as a generalization of the standard Bell inequalities (see Ref.~\cite{Branciard13} for a similar formulation in the entanglement scenario), but also be used to generalize the result of Ref.~\cite{Kocsis15}, wherein the family of two-qubits Werner state is explicitly considered. Additionally, we prove that, for any given steerable assemblage, there always exists a variant of a QRSG $\beta=\{\beta_{a,1}^{x,y}\}$ such that the corresponding average payoff in Eq.~\eqref{payoff2} is strictly higher than those caused by unsteerable ones. The proof is given in Section A of the Supplementary Material~\cite{suppl}.




Now, we stand in the position to introduce the MDI-SM for an unknown assemblage $\{\sigma_{a|x}\}$, denoted by
\begin{equation}
\SMDI:=\max\left\{\SOMDI-1,0\right\},
\label{Eq_SMDI}
\end{equation}
with
\begin{equation}
\SOMDI := \sup_{\beta,\mathbf{P}}\frac{I(\mathbf{P},\beta)}{I_{\text{LHS}}(\beta)},
\label{Eq_SOMDI}
\end{equation}
where $I_{\text{LHS}}(\beta)=\sup_{\mathbf{P}\in\mathsf{LHS'}}I(\mathbf{P},\beta)$ is the maximal payoff from unsteerable assemblages
for a given variant QRSG $\beta$. Importantly, the supremum over $\beta$ and $\mathbf{P}$ in Eq.~\eqref{Eq_SOMDI} are independently performed by the referee and Bob, respectively.
The former needs the referee to choose the optimal $\beta$ satisfying the spanned relation
\begin{equation}
F_{a|x} = \sum_y \beta_{a,1}^{x,y}\omega_y^{\intercal}
\label{Eq_span_relation}
\end{equation}
for positive semidefinite operators $\{F_{a|x}\geq 0\}$ since $\{\omega_y\}$ forms a tomography complete set.
On the other hand, Bob's optimization is carried out by subtly choosing a proper measurement, described by POVMs $E_1^{BB_{0}}$, and $\openone-E_1^{BB_{0}}\}$.
With Eq.~\eqref{Eq_span_relation}, Eq.~\eqref{Eq_SOMDI} can be reformulated as
\begin{equation}
\begin{aligned}
&\SOMDI=\\
& \sup_{\mathbf{F}\geq 0, E_1}\frac{\sum_{a,x}\tr\left[\EBB(\sigma_{a|x}\otimes F_{a|x}^\intercal)\right]}{\sup_{\tau\in\mathsf{LHS}}~ \sum_{a,x}\tr\left[\EBB(\tau_{a|x}\otimes F_{a|x}^\intercal)\right]},
\end{aligned}
\label{Eq_SOMDI2}
\end{equation}
where $\mathbf{F}$, $E_1$, and $\tau$, respectively, denote $\{F_{a|x}\}$, $E_1^{BB_0}$, and $\{\tau_{a|x}\}$ for brevity.

The optimization problem over Bob's measurement operator $E_1^{BB_{0}}$ is addressed in the following Lemma by resorting the projection onto the maximally entangled state.
\begin{lemma}\label{lemma:opt_measure}
The supremum over $E_1$ in Eq.~\eqref{Eq_SOMDI2} is always achieved if $\EBB$ is the projection onto the maximally entangled state, i.e., $\EBB=|\PhiBB\rangle\langle\PhiBB|$, with $|\PhiBB\rangle=1/\sqrt{d_{B}}\sum_{i=1}^{d_{B}}|i\rangle\otimes|i\rangle$. Moreover, it is independent of the chosen tomographically complete set $\{\omega_y\}$.
\end{lemma}
The proof is given in Section B of the Supplementary Material~\cite{suppl}. With the help of Lemma~\ref{lemma:opt_measure}, we arrive at the main result of this work below:
\begin{theorem}\label{theorem:MDI_to_SF}
The proposed MDI-SM $\SMDI$ in Eq.~(\ref{Eq_SMDI}) is a standard measure of steerability, i.e., a convex steering monotone~\cite{Gallego15}, due to the equivalence to the steering fraction~\cite{Hsieh16} as well as the steering robustness~\cite{Piani15}.
\end{theorem}

\begin{proof}
It is easy to verify that $\SOMDI$ can achieve the steering fraction \cite{Hsieh16} when considering $\EBB$ in Eq.~\eqref{Eq_SOMDI2} to be the projection onto the maximally entangled state $\ket{\PhiBB}$, i.e.,
\begin{equation}
\begin{aligned}
&\SOMDI\Big|_{\EBB=\ket{\PhiBB}\bra{\PhiBB}}\\
&=\sup_{\mathbf{F}\geq 0}\frac{\sum_{a,x}\bra{\PhiBB}\sigma_{a|x}\otimes F_{a|x}^\intercal\ket{\PhiBB}}{\sup_{\tau\in\mathsf{LHS}}~ \sum_{a,x}\bra{\PhiBB}\tau_{a|x}\otimes F_{a|x}^\intercal\ket{\PhiBB}}\\
&=\sup_{\mathbf{F}\geq 0}\frac{\sum_{a,x}\tr\left[\sigma_{a|x} F_{a|x} \right]}{\sup_{\tau\in\mathsf{LHS}}\sum_{a,x}\tr\left[\tau_{a|x}F_{a|x}\right]}.
\end{aligned}
\label{Eq_MDISMandSF}
\end{equation}
The last quantity is exactly the steering fraction in Ref.~\cite{Hsieh16}.
We leave the proof of the equivalence between the steering fraction and the steering robustness in Section C of the Supplementary Material~\cite{suppl}.
\end{proof}

We have explicitly shown how to achieve the optimal $\mathbf{P}$ in Eq.~\eqref{Eq_SOMDI} with Lemma~\ref{lemma:opt_measure}.
However, it is not straightforward to obtain a general form of the optimal variant of QRSGs, i.e., $\{\beta_{a,1}^{x,y}\}$ in Eq.~\eqref{Eq_SOMDI}. In Section D of the Supplementary Material~\cite{suppl}, we provide an algorithmic method to construct a valid set $\{\beta_{a,1}^{x,y}\}$. The idea is to consider a target assemblage to be measured. The optimal steering witness $\{F_{a|x}\}$ can be obtained by the dual form of the semidefinite program of the steering robustness. Then, a valid set $\{\beta_{a,1}^{x,y}\}$ can be chosen by the spanned relation~\eqref{Eq_span_relation}. We note that obtaining an optimal \emph{semi-quantum nonlocal game} is in general a hard problem~\cite{Francesco12,Shahandeh17}, and the compromising way is that one makes some assumptions on the entanglement structure of the underlying state~\cite{Shahandeh17}. However, our result shows that, in the variant of QRSGs, there is no assumption on the structure of steerability of the underlying assemblage to obtain an optimal $\{\beta_{a,1}^{x,y}\}$.


Now, we would like to show that the MDI-SM is robust against detection losses. To see this, we consider the loss rate of Bob's measurement $\eta\in[0,1]$. The observed correlation in this case is $p_\eta(a,b|x,\tau_y)=\eta\cdot p(a,b|x,\tau_y)$, shrinking the MDI-SM by $\eta$, i.e., $\eta\cdot\SMDI$. As one can see, the MDI-SM is able to detect steerability in a MDI scenario with arbitrary detection losses and provide a lower bound on the steerability of the underlying assemblage (see Ref.~\cite{Verbanis16,Branciard13} for similar discussions in the MDI entanglement scenario.)

\emph{Example.---}In the following, we give an explicit example to analytically compute the proposed MDI-SM. We consider that Alice and Bob share the family of two-qubit Werner states, namely
\begin{equation}
\rho_{AB}^{v}=v\ket{\Phi^{-}}\bra{\Phi^{-}}+(1-v)\openone/4\quad~~ v\in [0,1],
\end{equation}
where $\ket{\Phi^{-}}=(\ket{10}-\ket{01})/\sqrt{2}$ is the singlet state and $v$ is the visibility ($0\leq v \leq 1$).
We consider the simplest case where Alice receives two classical inputs $x\in\{1,2\}$ from the referee. She performs two measurement settings in the bases of Pauli $X$ and $Z$. These two measurements create an assemblage with the maximum steerability in the scenario of Alice holding two measurment settings~\cite{Skrzypczyk14}, and the underlying assemblage is steerable as $v> 1/\sqrt{2}$.
Obtaining an assemblage through such a setting, one can obtain an optimal steering witness
\begin{equation}
F_{a|x}=\frac{1}{2+\sqrt{2}}[\openone+(-1)^a\sigma_x]\quad\forall a,x,
\end{equation}
by solving the optimazation problem (see Section E of the Supplementary Material~\cite{suppl} for the derivation), where $\sigma_{x=1}=X$ and $\sigma_{x=2}=Z$.

On the other hand, the tomographyically complete set $\{\omega_y\}$ Bob receives can be chosen as the eigenstates of the three Pauli matrices.
Then, through the spanned relation $F_{a|x}=\sum_y \beta_{a,1}^{x,y}\omega_y^{\intercal}$, a valid optimal set $\{\beta^{x,y}_{a,1}\}$ can be chosen as $\beta^{x,y}_{a,1}=2/(2+\sqrt{2})$ for $(a,x,y)=(1,1,1), (2,1,2), (1,2,3), (2,2,4)$, and $\beta^{x,y}_{a,1}=0$ otherwise.
By projecting Bob's joint systems onto the maximally entangled state $\ket{\Phi^{+}}=(\ket{11}+\ket{00})/\sqrt{2}$, the set of probability distributions $\{p(a,b|x,\omega_y)\}$ as well as the MDI-SM $\SMDI$ are obtained. The result of MDI-SM for the family of two-qubit Werner states is shown in Fig.~\ref{Werner_states}. Note that there are different ways to choose the set $\beta$, as long as the spanned relation is satisfied. We show another possible optimal set in Table I of Section E of the Supplementary Material~\cite{suppl}.

\begin{figure}[tbp]
\includegraphics[width=1\columnwidth]{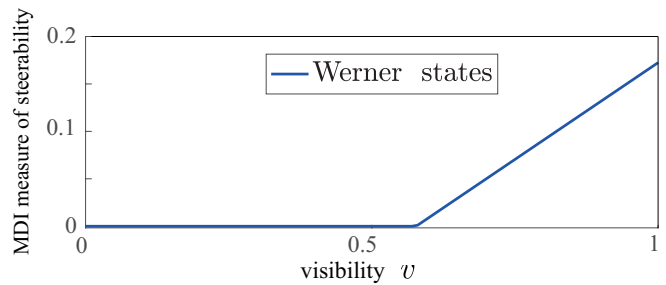}
\caption{The MDI-SMs $\SMDI$ of the assemblage of the two-qubit Werner state vesus the visibility $v$ by considering the scenario $(\mathcal{X},\mathcal{Y},\mathcal{A},\mathcal{B})=(2,6,2,2)$. The measure detects the steerability of the Werner states when the visibility $v\geq 1/\sqrt{2}$, which is exactly the bound of the visibility of the two-qubit Werner state with two projective-measurement settings on Alice's side. As shown in the main text, the value of the MDI-SM is equivalent to the steering robustness.
}
\label{Werner_states} %
\end{figure}

\emph{Concluding Remarks.---}In this work, we consider a variant of quantum refereed steering games (QRSGs), by which we introduce a measure of steerability in a measurement-device-independent (MDI) scenario, i.e., without making assumptions on the involved measurements nor the underlying assemblage. The only characterized quantities are the observed statistics and a tomographycially complete set of quantum states for Bob. Through this, all steerable assemblages can be witnessed, in contrast to the fact that only a subset of steerable assemblages can be detected in the standard device-independent (DI) scenario. 
We further show that it is a convex steering monotone by proving the equivalence with the steering fraction as well as the steering robustness. 
Therefore, the MDI-SM can be used to estimate the degree of entanglement of the unknown quantum state and measurement incompatibility of the involved measurements~\cite{Cavalcanti16}.
To our knowledge, this is the first work, which not only provides a \emph{computable measure of steerability based only on the observed statistics but also estimates the degree of measurement incompatibility of the involved measurements in a MDI scenario.} Additionally, our approach is able to detect steerability in a MDI scenario with arbitrary detection losses and provide a lower bound on the steerability of the underlying assemblage.

Moreover, we tackle the two optimization problems in Eq.~\eqref{Eq_SOMDI}, or equivalently, find the optimal strategies for the variant of QRSGs. The first is the problem of obtaining the general form of the optimal measurement. At the first glance, it seems to be a difficult problem since Bob has to optimize over all possible measurements. However, we show that the projection onto the maximally entangled state is always an optimal one for any steerable resource. 
Very recently, an arbitrary two-particle high-dimensional Bell state measurement has been proposed~\cite{Zhang19}, and it is expected that a high-dimensional MDI steering measure can be experimentally implemented with linear optics using current technology.
The second is the problem of obtaining the optimal game. Unlike other two types of generalized nonlocal games, i.e., semi-quantum nonlocal games~\cite{Francesco12} and the standard QRSGs~\cite{Cavalcanti13}, where the general optimal game for a given state is hard to be formulated~\cite{Shahandeh17}, the optimal game of the variant of QRSGs, which we consider, is easy to find.

This work also reveals some open questions: To calculate the value of the MDI-SM, or to obtain a MDI steering witness for an unknown steerable assemblage, can one directly estimate an optimal set of coefficients $\beta$, instead of obtaining it through the standard steering witness? (such as the approach used in~\cite{Verbanis16,Rosset18b}). 
Can we generally estimate the degree of steerability or entanglement when considering the effects of imperfections on the state preparation~\cite{Kocsis15,Jeon19,Rosset18b}?
It is also interesting to investigate whether our method can be modified to detect or measure all steerable assemblages in a standard DI scenario with the novel approach recently proposed in Ref,~\cite{Bowles18}. Since the formulation of the standard steering scenario can be applied to certify the security of the quantum keys~\cite{Branciard12}, one can ask if it is also the case in the MDI scenario.

H.-Y.K. and S.-L.C. contribute equally to this work. The authors acknowledge fruitful discussions with Francesco Buscemi, Ana Cristina Sprotte Costa, Yeong-Cherng Liang, Jeng-Dong Lin, Chau Nguyen, Paul Skrzypczyk, Adam Miranowicz, and Roope Uola. H.-Y.K. acknowledges the support of the Graduate Student Study Abroad Program (Grant No.~MOST 107-2917-I-006-002). S.-L.C. acknowledges the host by the group of theoretical quantum optics at the University of Siegen and the support from Postdoctoral Research Abroad Program (Grant No. MOST 107-2917-I-564-007) and from Deutsche Forschungsgemeinschaft (DFG, German Research Foundation) under Germany's Excellence Strategy -- EXC-2111--390814868. This work is supported partially by the National Center for Theoretical Sciences and Ministry of Science and Technology, Taiwan, Grants No.~MOST 107-2628-M-006-002-MY3, MOST 107-2811-M-006-017, and MOST 107-2627-E-006-001, and the Army Research Office (Grant No.~W911NF-19-1-0081). F.N. is supported in part by the MURI Center for Dynamic Magneto-Optics via the Air Force Office of Scientific Research (AFOSR) (FA9550-14-1-0040), Army Research Office (ARO) (Grant No.~W911NF-18-1-0358), Asian Office of Aerospace Research and Development (AOARD) (Grant No.~FA2386-18-1-4045), Japan Science and Technology Agency (JST) (Q-LEAP program, ImPACT program, and CREST Grant No. JPMJCR1676), Japan Society for the Promotion of Science (JSPS) (JSPS-RFBR Grant No. 17-52-50023, and JSPS-FWO Grant No. VS.059.18N), RIKEN-AIST Challenge Research Fund, and the John Templeton Foundation.

\emph{Note added:} after this work was submitted, experiments (arXiv: 1901.08298~\cite{Zhao2019}) have already verified the main prediction of this work.

\bibliography{ellipsoid_refs}

\clearpage
\appendix
\onecolumngrid
\begin{center}
{\bf \large Supplemental material}
\end{center}
\twocolumngrid

\section{MDI witnesses for all steerable assemblages}
For textural completeness, we first recall the standard steering witness. The set of all unsteerable assemblages $\mathsf{LHS}$ forms a convex set~\cite{Gallego15}. Therefore, for a given steerable assemblage $\{\sigma_{a|x}^{\rm{S}}\}$, there always exists a set of positive semidefinite operators $\{F_{a|x}\geq 0\}$, called a \emph{steering witness} $SW$, such that $\tr\sum_{a,x}F_{a|x}\,\sigma_{a|x}^{\rm{S}} > \alpha:= \max_{\{\sigma_{a|x}^{\rm{US}}\}\in\mathsf{LHS}} \tr \sum_{a,x} F_{a|x}\,\sigma_{a|x}^{\rm{US}}$, while $\tr\sum_{a,x}F_{a|x}\,\sigma_{a|x}^{\rm{US}} \leq \alpha \quad \forall \{\sigma_{a|x}^{\rm{US}}\}\in\mathsf{LHS}$.~\cite{Cavalcanti09,Pusey13,SDPreview17,Skrzypczyk14,Piani15}

The two conditions can be reformulated as follows:
\begin{equation}
\tr\sum_{a,x}\left(F_{a|x}-\frac{\alpha}{|\mathcal{X}|}\openone\right)\sigma_{a|x}^{\rm{S}}> 0,
\label{Eq_Sep_Th_1}
\end{equation}
while
\begin{equation}
\tr\sum_{a,x}\left(F_{a|x}-\frac{\alpha}{|\mathcal{X}|}\openone\right)\sigma_{a|x}^{\rm{US}}\leq 0 \quad\forall \{\sigma_{a|x}^{\rm{US}}\}\in\mathsf{LHS},
\label{Eq_Sep_Th_2}
\end{equation}
where $|\mathcal{X}|$ denotes the number of elements in $\mathcal{X}$, i.e., the number of the measurement settings.

Motivated by the result from Refs.~\cite{Branciard13,Kocsis15,Cavalcanti13}, here we show how to systematically construct a collection of steering witnesses in a MDI scheme, dubbed MDI-SWs. It is MDI since we certify steerability based only on the statistics $\{p(a,b|x,\omega_y)\}$ and on the fact that $\{\omega_y\}$ is a tomographically complete set. In what follows, we would like to address the problem under the framework of the resource theory of steering~\cite{Gallego15}, i.e., we will certify steerability of the underlying assemblage $\{\sigma_{a|x}\}$ instead of the quantum state $\rho_{AB}$.

Under the framework of the resource theory of steering~\cite{Gallego15}, the correlation is obtained from Bob's joint measurement on the assemblage, i.e., $p(a,1|x,\omega_y)=\text{tr}(\EBB \sigma_{a|x}\otimes\omega_y)$. The average payoff of an assemblage can then be defined as
\begin{equation}
I\Big(\mathbf{P}(\{\sigma_{a|x}\}),\beta\Big)=\sum_{a,x,y}\beta^{x,y}_{a,1}\,p(a,1|x,\omega_y).
\label{payoff3}
\end{equation}
where $\mathbf{P}(\{\sigma_{a|x}\}):=\{p(a,1|x,\omega_y)\}$.

Now we show that for any given steerable assemblage, one can properly choose a set of coefficients $\beta:=\{\beta_{a,1}^{x,y}\}$, such that $I\Big(\mathbf{P}(\{\sigma_{a|x}\}),\beta\Big)$ is a steering witness of the steerable assemblage. That is,
\begin{equation}
\begin{aligned}
\text{given}\quad & \{\sigma_{a|x}\}\notin\mathsf{LHS},\quad\exists \beta:=\{\beta_{a,1}^{x,y}\}\\
\text{such that}\quad & I(\{\sigma_{a|x}\},\beta)>0,\\
& I(\{\sigma_{a|x}^{\text{US}}\},\beta)\leq 0 \quad\forall \{\sigma_{a|x}^{\text{US}}\} \in\mathsf{LHS}.
\end{aligned}
\label{Eq_exist_MDISW}
\end{equation}

\begin{proof}
Since the set of Bob's input quantum states $\{\omega_{y}\}$ is a tomographically complete set, it can be used to span all Hermitian matrices of the same dimension:
\begin{equation}
\begin{aligned}
\text{given}&\quad\{F_{a|x}\}~\&~\{\omega_y\},\quad\exists\{\beta_{a,1}^{x,y}\}\\
\text{such that}&\quad F_{a|x}-\frac{\alpha}{|\mathcal{X}|}\openone=\sum_{y}\beta^{x,y}_{a,1}\,\omega_{y}^\intercal~~\forall a,x,
\end{aligned}
\label{Eq_Span_Fax}
\end{equation}
where $\{F_{a|x}\}$ is a SW of the assemblage and $\beta_{a,1}^{x,y}$ is a set of some real numbers. The transposition $\intercal$ is for convenience, as will be shown later.

(i) First we prove the second requirement of Eq.~\eqref{Eq_exist_MDISW}. Each component in the correlation $\{p(a,1|x,\omega_y)\}$ admitting a LHS model can be expressed as
\begin{equation}
\begin{aligned}
p(a,1|x,\omega_y)&=
\tr\left[\EBB(\sigma_{a|x}\otimes \omega_{y})  \right]\\
&=\sum_{\lambda}p(\lambda)\,p(a|x,\lambda)\tr\left[(\tilde{E}^{{ B}_0}_{1,\lambda}\omega_{y})  \right],
\end{aligned}
\end{equation}
where $\tilde{E}^{B_0}_{1,\lambda}:=\tr_{B}[\EBB(\sigma_{\lambda}\otimes\openone)]$ is an effective POVM element. The payoff of the assemblage is then written as
\begin{equation}
\begin{aligned}
&I(\{\sigma_{a|x}\},\beta) := \sum_{a,x,y}\beta^{x,y}_{a,1}p(a,1|x,\omega_y)\\
&=\sum_{a,x,\lambda}p(\lambda)p(a|x,\lambda)\tr\left[\tilde{E}^{{ B}_0}_{1,\lambda}\left(\sum_{y}\beta^{x,y}_{a,1}\omega_{y}\right)  \right]\\
&=\tr\left[\sum_{a,x}\left(F_{a|x}-\frac{\alpha}{|\mathcal{X}|}\openone\right)\sum_\lambda p(\lambda)p(a|x,\lambda) (\tilde{E}^{{ B}_0}_{1,\lambda})^\intercal \right]\\
&\leq 0,
\end{aligned}
\end{equation}
where the inequality holds due to Eq.~\eqref{Eq_Sep_Th_2}.

(ii) Now we prove the first requirement of Eq.~\eqref{Eq_exist_MDISW}. We choose the joint measurement performed by Bob to be the projection onto the maximally entangled state $\ket{\PhiBB}=1/\sqrt{d_{B}}\sum_{i=1}^{d_{B}}\ket{i}\otimes \ket{i}$. Therefore, each component of the correlation can be expressed as
\begin{equation}
\begin{aligned}
p(a,1|x,\omega_y)&=\tr\left[\EBB(\sigma_{a|x}\otimes \omega_{y})  \right]\\
&=\tr\left[( \ket{\PhiBB}\bra{\PhiBB})(\sigma_{a|x}\otimes \omega_{y})  \right]\\
&=\tr\left[\omega_{y}^\intercal~\sigma_{a|x}  \right]/d_{B}.
\end{aligned}
\end{equation}
The average payoff is reformulated as
\begin{equation}
\begin{aligned}
&I(\{\sigma_{a|x}\},\beta):=\sum_{a,x,y}\beta^{x,y}_{a,1}\,p(a,1|x,\omega_y)\\
&=\sum_{a,x}\tr\left[\left(\sum_{y}\beta^{x,y}_{a,1}\omega_{y}^\intercal\right)\sigma_{a|x}^{\text{S}}   \right]/d_{B}\\
&=\sum_{a,x}\tr\left[\left(F_{a|x}- \frac{\alpha}{|\mathcal{X}|}\openone \right) \sigma_{a|x}^{\text{S}}\right]/d_{\text{\tiny B}}>0,
\end{aligned}
\end{equation}
where the inequality holds according to Eq.~\eqref{Eq_Sep_Th_1}.
\end{proof}

\section{The equivalence between the MDI measure of steerability and the steering fraction}
Let us now rewrite the definition of the MDI steering measure (MDI-SM), i.e., Eq.~\eqref{Eq_SOMDI} in the main text
\begin{equation}
\SMDI:=\max\left\{\SOMDI-1,0\right\},
\label{Eq_Eq_appen_SOMDI2_SMDI}
\end{equation}
where
\begin{equation}
\SOMDI := \sup_{\beta,\mathbf{P}}\frac{\sum_{a,x,y}\beta^{x,y}_{a,1}\,p(a,1|x,\omega_y)}{\sup_{\bar{\mathbf{P}}\in\mathsf{LHS'}} ~\sum_{a,x,y}\beta^{x,y}_{a,1}\,\bar{p}(a,1|x,\omega_y)}.
\label{Eq_appen_SOMDI}
\end{equation}
By replacing $p(a,b|x,\omega_y)$ with $\tr[\EBB \sigma_{a|x}\otimes\omega_y]$ and using the spanned relation $F_{a|x} = \sum_y \beta_{a,1}^{x,y}\,\omega_y^\intercal$, then $\SOMDI$ can be reformulated as [i.e., Eq.~\eqref{Eq_SOMDI2}]:
\begin{equation}
\begin{aligned}
&\SOMDI\\
&= \sup_{\mathbf{F}\geq 0, E_1}\frac{\sum_{a,x}\tr\left[\EBB(\sigma_{a|x}\otimes F_{a|x}^\intercal)\right]}{\sup_{\tau\in\mathsf{LHS}}~ \sum_{a,x}\tr\left[\EBB(\tau_{a|x}\otimes F_{a|x}^\intercal)\right]},
\end{aligned}
\label{Eq_appen_SOMDI2}
\end{equation}
where $\mathbf{F}\geq 0$ denotes $\{F_{a|x}\geq 0\}$ for brevity. Since $E_1^{BB_0}$ is a POVM element, it is diagonalizable and can be taken as a linear combination of rank-$1$ projectors with coefficients lying between 0 and 1. Since any rank-$k$ projector can be produced by acting a separable operation on the maximally entangled state, $\EBB$ can be written as
\begin{equation}
\begin{aligned}
\EBB &= \sum_{k,i}u(k)\tilde{A}^{k}_{i}\otimes \tilde{B}^{k}_{i}\ket{\Phi}\bra{\Phi}\tilde{A}^{k\dagger}_{i}\otimes \tilde{B}^{k\dagger}_{i},\\
&= \sum_{k,i}A^{k}_{i}\otimes B^{k}_{i}\ket{\Phi}\bra{\Phi}A^{k\dagger}_{i}\otimes B^{k\dagger}_{i},\\
\end{aligned}
\end{equation}
where $u(k)$ denotes the coefficients between $0$ and $1$, $A_{i}^{k}\otimes B_{i}^{k}=\sqrt{u(k)}\tilde{A}^{k}_{i}\,\otimes \tilde{B}^{k}_{i}$ is the redefined Kraus operators for each $i$, and (for brevity) $\ket{\Phi}$ denotes $\ket{\PhiBB}$. Then, we can proceed to write $\SOMDI$ as
\onecolumngrid
\vspace{\columnsep}
\begin{center}
\begin{equation}
\begin{aligned}
&\sup_{\mathbf{F}\geq 0,A^{k}_{i},B^k_i}\frac{\sum_{a,x,k,i}\tr\left[A^{k}_{i}\otimes B^{k}_{i}\ket{\Phi}\bra{\Phi}A^{k\dagger}_{i}\otimes B^{k\dagger}_{i}(\sigma_{a|x}\otimes F_{a|x}^{\intercal})\right]}{\sup_{\tau\in\mathsf{LHS}}\sum_{a,x,k,i}\tr\left[A^{k}_{i}\otimes B^{k}_{i}\ket{\Phi}\bra{\Phi}A^{k\dagger}_{i}\otimes B^{k\dagger}_{i}(\tau_{a|x}\otimes F_{a|x}^{\intercal})\right]}\\
&=\sup_{\mathbf{F}\geq 0,A^{k}_{i},B^k_i}\frac{\sum_{a,x,k,i}\bra{\Phi}A^{k\dagger}_{i}\otimes B^{k\dagger}_{i}(\sigma_{a|x}\otimes F_{a|x}^{\intercal})A^{k}_{i}\otimes B^{k}_{i}\ket{\Phi}}{\sup_{\tau\in\mathsf{LHS}}\sum_{a,x,k,i}\bra{\Phi}A^{k\dagger}_{i}\otimes B^{k\dagger}_{i}(\tau_{a|x}\otimes F_{a|x}^{\intercal})A^{k}_{i}\otimes B^{k}_{i}\ket{\Phi}}\\
&=\sup_{\mathbf{F}\geq 0,A^{k}_{i},B^k_i}\frac{\sum_{a,x,k,i}\bra{\Phi}(A^{k\dagger}_{i}\sigma_{a|x}A^{k}_{i})\otimes (B^{k\dagger}_{i} F_{a|x}^{\intercal}B^{k}_{i})\ket{\Phi}}{\sup_{\tau\in\mathsf{LHS}}\sum_{a,x,k,i}\bra{\Phi}(A^{k\dagger}_{i}\tau_{a|x}A^{k}_{i})\otimes( B^{k\dagger}_{i} F_{a|x}^{\intercal}B^{k}_{i})\ket{\Phi}}\\
&=\sup_{\mathbf{F}\geq 0,A^{k}_{i},B^k_i}\frac{\sum_{a,x,k,i}\tr\left[A^{k\dagger}_{i}\sigma_{a|x}A^{k}_{i}\cdot B^{k\intercal}_{i} F_{a|x}B^{k\dagger^\intercal}_{i}\right]}{\sup_{\tau\in\mathsf{LHS}}\sum_{a,x,k,i}\tr\left[A^{k\dagger}_{i}\tau_{a|x}A^{k}_{i}\cdot B^{k\intercal}_{i} F_{a|x}B^{k\dagger^\intercal}_{i}\right]}\\
&=\sup_{\mathbf{F}\geq 0,A^{k}_{i},B^k_i}\frac{\sum_{a,x}\tr\left[\sigma_{a|x}\sum_{k,i}A^{k}_{i} B^{k\intercal}_{i} F_{a|x}B^{k\dagger^\intercal}_{i}A^{k\dagger}_{i}\right]}{\sup_{\tau\in\mathsf{LHS}}\sum_{a,x}\tr\left[\tau_{a|x}\sum_{k,i}A^{k}_{i} B^{k\intercal}_{i} F_{a|x}B^{k\dagger^\intercal}_{i}A^{k\dagger}_{i}\right]}\\
&\leq\sup_{\mathbf{F}\geq 0}\frac{\sum_{a,x}\tr\left[\sigma_{a|x} F_{a|x} \right]}{\sup_{\tau\in\mathsf{LHS}}\sum_{a,x}\tr\left[\tau_{a|x}F_{a|x}\right]}.
\end{aligned}
\label{Eq_proof_SF}
\end{equation}
\end{center}
\vspace{\columnsep}

\twocolumngrid
The inequality is due to the fact that the convex set $\mathbf{F}$ is a superset of the one after performing the completely positive map, i.e., $\mathbf{F}':=\{\sum_{k,i}A^{k}_{i} B^{k\intercal}_{i} F_{a|x}B^{k\dagger^\intercal}_{i}A^{k\dagger}_{i}\}_{a,x}$. The last quantity is exactly the steering fraction proposed by~\cite{Hsieh16}. From the result of the next section, we obtain that $\SMDI$ is also the same as the steering robustness.


\section{Proof of the equivalence between the steering fraction and the steering robustness}
In this section, we explicitly prove the equivalence between the steering fraction and  the steering robustness, although their equivalence is implicitly mentioned in some references (see, e.g., Ref.~\cite{SDPreview17}). The steering robustness of a given assemblage can be obtained by the dual program described in Eq.~\eqref{Eq_Fax_SR}. On the other hand, the steering fraction ($S_F$) of the given assemblage is defined as~\cite{Hsieh16}
\begin{equation}
\begin{aligned}
S_F + 1 = \max_{\mathbf{F}\geq 0}\quad \frac{\tr\sum_{a,x}F_{a|x}\,\sigma_{a|x}}{\max_{\tau\in\mathsf{LHS}}\tr\sum_{a|x}F_{a|x}\,\tau_{a|x}}.
\end{aligned}
\label{Eq_SF}
\end{equation}
We can rewrite it as
\begin{equation}
S_F +1 = \max_{\tilde{\mathbf{F}}\geq 0}\quad \tr\sum_{a,x}\tilde{F}_{a|x}\,\sigma_{a|x},
\end{equation}
where
\begin{equation}
\tilde{F}_{a|x}:=\frac{F_{a|x}}{\max_{\tau\in\mathsf{LHS}}\tr\sum_{a|x}F_{a|x}\,\tau_{a|x}}\geq 0.
\end{equation}
Therefore, to prove the equivalence between Eqs.~\eqref{Eq_Fax_SR} and \eqref{Eq_SF}, it is equivalent to prove
\begin{equation}
\sum_{a,x} D(a|x,\lambda)\tilde{F}_{a|x}\leq\openone\quad\forall \lambda.
\end{equation}
\begin{proof}
For each $\lambda$, the quantity $\openone - \sum_{a,x} D(a|x,\lambda)\tilde{F}_{a|x}$ is multiplied by a subnormalized quantum state $\rho_\lambda\geq 0$. We take the trace, and sum over all $\lambda$:
\begin{equation}
\begin{aligned}
&\tr\sum_\lambda\left(\openone-\frac{\sum_{a,x}D(a|x,\lambda)F_{a|x}}{\max_{\tau\in\mathsf{LHS}}\tr\sum_{a|x}F_{a|x}\tau_{a|x}}\right)\rho_\lambda\\
&=1-\frac{\tr\sum_{a,x}F_{a|x}\sigma_{a|x}^{\text{US}}}{\max_{\tau\in\mathsf{LHS}}\tr\sum_{a|x}F_{a|x}\tau_{a|x}},
\end{aligned}
\end{equation}
which is non-negative for all $\rho_\lambda\geq 0$ and $\lambda$. Since the only constraint between the free parameters $\rho_\lambda$ is $\tr\sum_\lambda\rho_\lambda=1$, we derive this condition
\begin{equation}
\openone - \frac{\sum_{a,x}D(a|x,\lambda)F_{a|x}}{\max_{\tau\in\mathsf{LHS}}\tr\sum_{a|x}F_{a|x}\tau_{a|x}} \geq 0\quad\forall\lambda.
\end{equation}
\end{proof}

\section{Construction of MDI-SWs and MDI-SMs from the standard steering witnesses}\label{Sec_method}
In this section, we provide an algorithmic method for constructing a set of coefficients $\beta$ of the MDI-SW and MDI-SM from the standard steering witness $\{F_{a|x}\}$. For a target steerable assemblage $\{\sigma_{a|x}\}$, one can construct a MDI-SW through the following steps:

\textbf{1.} Choose a tomographically complete set $\{\omega_y\}$ to be Bob's quantum inputs.

\textbf{2.} Consider the optimal standard steering witness $\{F_{a|x}\}$ of the target assemblage $\{\sigma_{a|x}\}$, which can be obtained either from the dual SDP program of the steerable weight $S_W$~\cite{Skrzypczyk14}
\begin{equation}
\begin{aligned}
\text{given}\quad &\{\sigma_{a|x}\}\\
S_W+1 = \min\quad &\tr\sum_{a,x}F_{a|x}\sigma_{a|x}\\
\text{such that}\quad &\sum_{a,x}D(a|x,\lambda)F_{a|x}\geq \openone\quad\forall\lambda\\
&F_{a|x}\geq 0\quad \forall a,x,
\end{aligned}
\label{Eq_Fax_SW}
\end{equation}
or from the dual SDP program of the steering robustness $S_R$~\cite{Piani15}:
\begin{equation}
\begin{aligned}
\text{given}\quad &\{\sigma_{a|x}\}\\
S_R+1 = \max\quad &\tr\sum_{a,x}F_{a|x}\sigma_{a|x}\\
\text{such that}\quad &\sum_{a,x}D(a|x,\lambda)F_{a|x}\leq \openone\quad\forall\lambda\\
&F_{a|x}\geq 0\quad \forall a,x.
\end{aligned}
\label{Eq_Fax_SR}
\end{equation}

\textbf{3.} Choose a set of coefficients $\beta:=\{\beta_{a,1}^{x,y}\}$ satisfying the spanned relation:
\begin{equation}
F_{a|x}-\frac{\openone}{|\mathcal{X}|} = \sum_y \beta_{a,1}^{x,y}\,\omega_y^\intercal\quad\forall a,x.
\end{equation}

\textbf{4.} Finally, $I(\mathbf{P}(\{\sigma_{a|x}\}),\{\beta\}) := \sum_{a,x,y}\beta^{x,y}_{a,1}p(a,1|x,\omega_y)$ is a MDI-SW. The negative value certifies the steerability if we consider the program Eq.~\eqref{Eq_Fax_SW} of the steerable weight in the second step, while the positive value certifies the steerability if we consider the steering robustness described by Eq.~\eqref{Eq_Fax_SR}.

To construct the MDI-SM of the given assemblage, we must follow these steps:

\textbf{1.} Choose a tomographically complete set $\{\omega_y\}$ to be Bob's quantum inputs.

\textbf{2.} Choose Bob's measurement to be in the basis $\{E_1^{BB_0},\openone-E_1^{BB_0}\}$, with $E_1^{BB_0}$ being the projection onto the maximally entangled state $(1/\sqrt{d})\sum_{ii}^d|ii\rangle$.

\textbf{3.} From the above two steps, one obtains the optimal correlation $\{p(a,1|x,\omega_y)=\tr(E_1^{BB_0}\sigma_{a|x}\otimes\omega_y)\}$.

\textbf{4.} Consider the optimal standard steering witness $\{F_{a|x}\}$ of the assemblage $\{\sigma_{a|x}\}$, which is obtained from the dual SDP program Eq.~\eqref{Eq_Fax_SR} of the steering robustness.

\textbf{5.} Choose a set of coefficients $\beta:=\{\beta_{a,1}^{x,y}\}$ satisfying the spanned relation:
\begin{equation}
F_{a|x} = \sum_y \beta_{a,1}^{x,y}\omega_y^\intercal\quad\forall a,x.
\end{equation}

\textbf{6.} Finally,
\begin{equation}
\max\Big\{\frac{\sum_{a,x,y}\beta^{x,y}_{a,1}p(a,1|x,\omega_y)}{\sup_{\bar{\mathbf{P}}\in\mathsf{LHS'}} ~\sum_{a,x,y}\beta^{x,y}_{a,1}\bar{p}(a,1|x,\omega_y)}-1,0\Big\}
\end{equation}
is the MDI-SM, where the denominator [see Eq.~\eqref{Eq_MDISMandSF}] is
\begin{equation}
\sup_{\bar{\mathbf{P}}\in\mathsf{LHS'}} ~\sum_{a,x,y}\beta^{x,y}_{a,1}\bar{p}(a,1|x,\omega_y)=\frac{1}{d}.
\end{equation}
One may find that the algorithmic method for constructing the MDI-SM is not genuine MDI since the assemblage has to be known. We have mentioned this in the last section of the main text, i.e., obtaining an optimal set $\beta$ in a MDI scenario is one of the open problems. We would also like to stress that the definition of the MDI-SM itself and the proof of the equivalence with the steering robustness are still in the MDI scheme.

\section{Analytical construction of the MDI steering measure for Werner states}

In this section, we provide an analytical construction of the MDI steering measure of an assemblage obtained from the two-qubit Werner state. The procedure is the same as the algorithmic method in Section~\ref{Sec_method} of this Supplementary material. To obtain the optimal standard steering witness $\{F_{a|x}\}$, we use a similar technique to the one used in Ref.~\cite{SDPreview17}. The two-qubit Werner state is written as
\begin{equation}
\rho_{AB}^{v}=v\ket{\Phi^{-}}\bra{\Phi^{-}}+(1-v)\openone/4,\quad v\in [0,1],
\end{equation}
where $\ket{\Phi^{-}}=(\ket{10}-\ket{01})/\sqrt{2}$ is the singlet state. We take the measurements performed by Alice to be in the bases of Pauli $X$ and $Z$. The corresponding assemblage is then given by~\cite{SDPreview17}
\begin{equation}
\sigma_{a|x}=v\frac{\openone+(-1)^{a+1}\hat{n}_x\cdot\vec{\sigma}}{2}+\frac{1-v}{4}\openone\quad\forall a,x,
\end{equation}
where $\hat{n}_{1}=(1,0,0)$ and $\hat{n}_{2}=(0,0,1)$ are vectors on the Bloch sphere, and $\vec{\sigma}=(X,Y,Z)$ is the set of Pauli matrices. Any two-dimensional Hermitian matrix $F_{a|x}$ can be expressed as $F_{a|x}=\gamma_{a|x}\openone+\vec{\kappa}_{a|x}\cdot\vec{\sigma}$, with $\gamma_{a|x}$ being a real number and $\vec{\kappa}_{a|x}$ being a three dimensional vector. Then, we arrive at
\begin{equation}
\begin{aligned}
\tr\sum_{a,x}F_{a|x}\sigma_{a|x}&=\frac{1}{2}(\gamma_{1|1}+\gamma_{2|1}+\gamma_{1|2}+\gamma_{2|2})\\&+\frac{v}{2}\left[\hat{n}_1\cdot(\vec{\kappa}_{1|1}-\vec{\kappa}_{2|1})+\hat{n}_2\cdot(\vec{\kappa}_{1|2}-\vec{\kappa}_{2|2})\right].
\end{aligned}
\end{equation}
The above quantity achieves its maximal value when $\vec{\kappa}_{a|x}$ are aligned or anti-aligned with $\hat{n}_{x}$. Specifically, we set $\vec{\kappa}_{1|1}\parallelsum \hat{n}_1$, $\vec{\kappa}_{2|1}\parallelsum -\hat{n}_1$, $\vec{\kappa}_{1|2}\parallelsum \hat{n}_2$, and $\vec{\kappa}_{2|2}\parallelsum -\hat{n}_2$. With the above conditions, we obtain
\begin{equation}
\begin{aligned}
\tr\sum_{a,x}F_{a|x}\sigma_{a|x}&=\frac{1}{2}(\gamma_{1|1}+\gamma_{2|1}+\gamma_{1|2}+\gamma_{2|2})\\&+\frac{v}{2}\left(\kappa_{1|1}+\kappa_{2|1}+\kappa_{1|2}+\kappa_{2|2}\right),
\end{aligned}
\label{alpha_W}
\end{equation}
where for each $a$ and $x$, $\kappa_{a|x}$ is a positive real number corresponding to the length of each vector $\vec{\kappa}_{a|x}$. In order to satisfy the second constraint in Eq.~\eqref{Eq_Fax_SR}, we have
\begin{equation}
\gamma_{a|x}\geq \kappa_{a|x}~~\forall~a,x.
\label{first_constraint}
\end{equation}
Considering all deterministic strategies $\lambda=(a_{x=1},a_{x=2})$ and the first constraint in Eq.~\eqref{Eq_Fax_SR}, we obtain the following inequalities:

\begin{table}
\centering  
\begin{tabular}{|c|c|c|c|c|c|c|c|} \hline
$\beta^{1,1}_{1,1}$ & $\beta^{2,1}_{1,1}$ & $\beta^{1,1}_{2,1}$& $\beta^{2,1}_{2,1}$&
$\beta^{1,2}_{1,1}$ & $\beta^{2,2}_{1,1}$ & $\beta^{1,2}_{2,1}$& $\beta^{2,2}_{2,1}$\\ \hline
$2\kappa$ & $\kappa$ & 0 & $\kappa$&0 & $\kappa$ & $2\kappa$ & $\kappa$\\ \hline
$\beta^{1,3}_{1,1}$ & $\beta^{2,3}_{1,1}$ & $\beta^{1,3}_{2,1}$& $\beta^{2,3}_{2,1}$&
$\beta^{1,4}_{1,1}$ & $\beta^{2,4}_{1,1}$ & $\beta^{1,4}_{2,1}$& $\beta^{2,4}_{2,1}$\\ \hline
0 & 0 & 0 & 0&0 & 0 & 0 & 0\\ \hline
$\beta^{1,5}_{1,1}$ & $\beta^{2,5}_{1,1}$ & $\beta^{1,5}_{2,1}$& $\beta^{2,5}_{2,1}$&
$\beta^{1,6}_{1,1}$ & $\beta^{2,6}_{1,1}$ & $\beta^{1,6}_{2,1}$& $\beta^{2,6}_{2,1}$\\ \hline
0 & $\kappa$ & 0 & $-\kappa$&0 & $-\kappa$ & 0 & $\kappa$\\ \hline
\end{tabular}
\caption{A choice of the set $\{\beta^{x,y}_{a,1}\}$ of the MDI-SM of the assemblage produced from the two-qubit Werner state.}
\label{table_beta}
\end{table}

\begin{equation}
\begin{aligned}
&(1-\gamma_{1|1}-\gamma_{1|2})\geq \parallel\vec{\kappa}_{1|1}+\vec{\kappa}_{1|2}\parallel=\sqrt{\kappa_{1|1}^2+\kappa_{1|2}^2},\\
&(1-\gamma_{1|1}-\gamma_{2|2})\geq \parallel\vec{\kappa}_{1|1}+\vec{\kappa}_{2|2}\parallel=\sqrt{\kappa_{1|1}^2+\kappa_{2|2}^2},\\
&(1-\gamma_{2|1}-\gamma_{1|2})\geq \parallel\vec{\kappa}_{2|1}+\vec{\kappa}_{1|2}\parallel=\sqrt{\kappa_{2|1}^2+\kappa_{1|2}^2},\\
&(1-\gamma_{2|1}-\gamma_{2|2})\geq \parallel\vec{\kappa}_{2|1}+\vec{\kappa}_{2|2}\parallel=\sqrt{\kappa_{2|1}^2+\kappa_{2|2}^2}.
\end{aligned}
\label{second_constraint}
\end{equation}
Combining Eqs.~\eqref{first_constraint} and \eqref{second_constraint}, the constraints can be reformulated as
\begin{equation}
\begin{aligned}
\sqrt{\kappa_{a_1|1}^2+\kappa_{a_2|2}^2}&\leq 1-\gamma_{a_1|1}-\gamma_{a_2|2}\\
&\leq 1-\kappa_{a_1|1}-\kappa_{a_2|2}~~\forall a_1,a_2\in\{1,2\}.
\end{aligned}
\label{final_constraint}
\end{equation}
From Eqs.~\eqref{alpha_W},~\eqref{first_constraint}, and \eqref{second_constraint}, we can see that $\gamma_{a|x}$ as well as $\kappa_{a|x}$ are permutation symmetrical to $a$ and $x$. Therefore, without loss of generality, we can assume that $\gamma_{a|x}=\gamma$ and $\kappa_{a|x}=\kappa$. Finally, from Eq.~\eqref{final_constraint} and the symmetric rule, we obtain the two inequalities as
\begin{equation}
\frac{1}{2+\sqrt{2}}\geq \kappa\geq 0~~\text{and}~~\frac{1-\sqrt{2\kappa^2}}{2}\geq \gamma\geq\kappa.
\label{gamma_beta}
\end{equation}
Since Eq.~\eqref{alpha_W} is a linear function of $\gamma$ and $\kappa$, the local maximal value takes place at the extremal points $\gamma=\kappa=(2+\sqrt{2})^{-1}$ of the constraint~\eqref{gamma_beta}. Therefore, the optimal $\{F_{a|x}\}$ in Eq.~\eqref{Eq_Fax_SR} is analytically constructed as
\begin{equation}
F_{a|x}=\frac{1}{2+\sqrt{2}}[\openone+(-1)^{a+1}\sigma_x]\quad~~\forall a,x.
\end{equation}

Now we take Bob's input quantum states $\{\omega_y\}$ to be the eigenstates of Pauli matrices, which form a tomographically complete set. The above steering functional can be spanned by this set:
\begin{equation}
F_{a|x}=\sum_y \beta^{x,y}_{a,1}\omega_y^\intercal\quad~~\forall a,x.
\end{equation}
Except a choice of the set $\{\beta^{x,y}_{a,1}\}$ shown in the main text, here we list an other feasible one in Table~\ref{table_beta}. The steerability of the assemblage created by the measurements on the two-qubit Werner state can then be obtained in a MDI scenario, which is shown in Fig.~\ref{fig_EPR_MDI} in the main text.

\end{document}